\documentclass[12pt]{article}
\usepackage{geometry}
\usepackage{authblk}
\usepackage{hyperref}
\usepackage{appendix}
\usepackage{graphicx, wrapfig}
\graphicspath{{C:/Users/Rem/Documents/Documents/PhD/Thesis/Impacts of a Policy Set on Wealth Distribution Dynamics/Images/}}
\usepackage{amssymb}
\usepackage{amsmath}
\usepackage{mathtools}
\usepackage{amsthm}
\usepackage{esvect}
\usepackage{caption}
\usepackage{color,soul}

\newtheorem{definition}{Definition}
\newtheorem{theorem}{Theorem}[section]
\newtheorem*{remark}{Remark}

\title{\Large \textbf{DeTEcT: Dynamic and Probabilistic Parameters Extension}\\ \large Modelling wealth distribution in token economies with time-dependent parameters}
\author{R. Sadykhov, Dr.G. Goodell and Prof.P. Treleaven}
\affil{University College London}
\date{February 2024}

\begin{document}

\begin{titlepage}
\maketitle
\thispagestyle{empty}
\begin{abstract}
This paper presents a theoretical extension of the DeTEcT framework proposed by Sadykhov et al. \cite{DeTEcT}, where a formal analysis framework was introduced for modelling wealth distribution in token economies. DeTEcT is a framework for analysing economic activity, simulating macroeconomic scenarios, and algorithmically setting policies in token economies. This paper proposes four ways of parametrizing the framework, where dynamic vs static parametrization is considered along with the probabilistic vs non-probabilistic. Using these parametrization techniques, we demonstrate that by adding restrictions to the framework it is possible to derive the existing wealth distribution models from DeTEcT.\par
In addition to exploring parametrization techniques, this paper studies how money supply in DeTEcT framework can be transformed to become dynamic, and how this change will affect the dynamics of wealth distribution. The motivation for studying dynamic money supply is that it enables DeTEcT to be applied to modelling token economies without maximum supply (i.e., Ethereum), and it adds constraints to the framework in the form of symmetries.
\end{abstract}
\end{titlepage}

\section{Introduction}
Token economies are economies that have a medium for valuation, transaction and value storage in the form of tokens (i.e., currency native to the economy). These tokens can be minted (i.e., new tokens are created) and burned (i.e., existing tokens are destroyed), while the study of tokens and token economies is generally referred to as \emph{Tokenomics}.\par
Acting under the assumption that tokens are a mechanism for storage of wealth and representation of the value of scarce resources, we can define \emph{tokenomics} as the study of the efficient allocation of wealth in a token economy \cite{ISOVocabularyEconomics, WealthOfNations}. Some of the questions that tokenomics aims to address are: how does an economic system provision and allocate scarce resources, how does the economy interact with the external economic systems and stimuli, what guides the behaviour of economic participants, and what is the ``efficiency'' of these processes?\par
In an attempt to answer these question, we have previously defined a formal analysis framework called DeTEcT \cite{DeTEcT}, which models interactions between groups of heterogeneous agents in token economies, simulates wealth distribution, and can be used to propose policies to be implemented in a token economy to achieve a desired wealth distribution. DeTEcT helps to define an economic system as a robust mathematical construction, facilitating theoretical study of the system, while maintaining flexibility to be adjusted to analyse economies with specific traits (e.g., economies with different rules of interactions between agents).\par
However, despite an already wide range of applications, the framework still has some limitations, such as constant maximum supply of tokens, exclusively static parameters, or the limitations of numerical solution method. In this paper we aim to address some of the limiting factors to make DeTEcT a more flexible framework, without loosing any mathematical robustness.\par
\subsection{Scope}
The aim of this paper is to define different possible configurations for the parametrization of the framework, demonstrate how existing wealth distribution models can be derived from DeTEcT, and define a method for modelling an economy with a dynamic maximum supply (i.e., maximum supply that varies in time). We address the first two questions in section \ref{sec:DerivingWealthModels} as these questions are interrelated, while we address the last question in section \ref{sec:DynamicMoneySupply}. Note that these sections are not arranged in a specific order, as we believe each of these questions is equally important.\par
We start section \ref{sec:DerivingWealthModels} by demonstrating how to parametrize the framework with parameters that change in time (i.e., dynamic parameters) and with parameters that are defined using probabilistic techniques (i.e., probabilistic parameters). The section concludes with us deriving the existing wealth distribution models from DeTEcT framework. The importance of this link is that the existing wealth distribution models we refer to have the support of empirical data, in particular the wealth distribution model with individual saving propensities \cite{MoneyInGasLikeMarkets}, as it results in a Pareto tail distribution which is observed in real-world economies \cite{HumanDevelopmentReport}.\par
Section \ref{sec:DynamicMoneySupply} expands the application of our framework to economies where maximum supply changes in time (e.g., Ethereum economy \cite{Ethereum}), and explores time translation symmetry in our framework.

\section{Related Work}
This paper expands on the mathematics of DeTEcT presented in Sadykhov at al. \cite{DeTEcT}, so the definitions used here are stated in that paper. In this section we will briefly reacquaint ourselves with the key concepts of DeTEcT, and cover some of the existing wealth distribution models from the literature.
\subsection{DeTEcT Framework}
\label{sec:RelatedWork:DeTEcT}
DeTEcT is a formal analysis framework that models wealth distribution in token economies. The fundamental building block for DeTEcT is clustering of agents in the economy into agent categories and categorizing possible interactions between different agent categories.\par
For example, assume an economy has an agent that manufactures a good (e.g., car) and sells it on to other agents (e.g., retail customers). We can create two agent categories in such an economy: \emph{Car Manufacturer} and \emph{Household}. We see that \emph{Car Manufacturer} interacts with \emph{Household} through the sale of cars. So, we can describe an interaction between these agent categories as \emph{Car Sale}.\par
Each agent is assumed to have some wealth, and has a mapping called individual wealth function, $f(a, t)=w$, where $a\in\Lambda$ is the label of the agent with $\Lambda$ being the set of all agents in the economy, $t$ is the time step, and $w$ is the wealth of the agent at time step $t$. Since each agent category is a set of agents, we have a wealth function $F(A,t)=W$ for each agent category $A$ at time step $t$, and it is a sum of all individual wealth functions of agents who are considered to be of agent category $A$ (i.e., $a\in A$) at time $t$.\par
However, we may have an economy where agents do not have a constant agent category and they can change between different agent categories. Therefore, the wealth of an agent category $A$ will change if agent $a$ now belongs to agent category $A^{\prime}$. We define this wealth reallocation from one agent category to another without the agent interacting with other agents as rotation.\par
Agent categories, interactions and rotations, together, are referred to as the tokenomic taxonomy, as it defines the possible routes for wealth movements. This allows us to create a complete set of possible wealth movements, and to construct the model for possible wealth reallocation between agent categories.\par
The motivation behind clustering agents into agent categories consists of two points: it allows for quicker processing as we model the macroeconomic state of the economy and therefore work with smaller parameter set, and that in most economies the policies are directed at groups of economic participants (e.g., corporate tax for companies vs income tax for households). Tokenomic taxonomy enables DeTEcT to be used for finding policies to be implemented in an economy, where the policies are directed at well-defined categories of agents.\par
Before we review DeTEcT framework, we must point out that we will be referring to distribution of tokens between agents in the economy as \emph{wealth distribution} rather than token distribution. The reason for this destinction is that DeTEcT can be applied to the modelling financial systems with different architectures, and is not limited to modelling only the token economies, despite the scope of this paper mainly concerning with the token economies. Additionally, in section \ref{sec:DerivingWealthModels} we discuss how DeTEcT links to existing wealth distribution models, and to make a smooth transition from our framework to existing works in this field, we refer to the distribution of tokens as wealth distribution, subject to the assumption that tokens in the closed economic systems act as a measure of wealth.\par
To summarize the review of DeTEcT, we define the compartmental dynamical system at the core of our framework. Assuming an economy has a well-defined maximum supply, we state that the sum of all wealth functions of all agent categories in the economy always equals to the maximum supply,
\begin{equation}
\label{eq:RelatedWork:ConstantMaxSupply}
	M = \sum_{A\in E_{t}}F(A,t) \quad \forall t\in T,
\end{equation}
where $M$ is the maximum supply, $E_{t}$ is the pseudo-partition of agent categories (i.e., the set of all agent categories), and $T$ is the set of all discrete time steps.\par
Note that ususally in token economies, there exists a mechanism that issues (i.e., mints) new tokens to agents, and a place where tokens are deposited when they are disposed of (i.e., burned). In the context of DeTEcT, we refer to these as separate agent categories denoted as $B\in E_{t}$ (i.e., Control Mechanism) and $D\in E_{t}$ (i.e., Token Dump) respectively. Control Mechanism is the required agent category for a token economy in DeTEcT since it is reponsible for initial token distribution and it is assumed to hold all unissued tokens (i.e., reserves). We refer to $M$ as the maximum supply as it is the total number of all tokens that could be issued (including the tokens that haven't been issued yet), while the concept of current circulating supply can be expressed in DeTEcT terms as $S(t)=M-F(B,t)$, where $S(t)$ is the circulating supply at some time $t\in T$.\par
From equation \ref{eq:RelatedWork:ConstantMaxSupply} we derive the conservation of wealth in time,
\begin{equation}
\label{eq:RelatedWork:ConservationPrinciple}
	0 = \sum_{A\in E_{t}}\frac{\Delta F(A,t)}{\Delta t}.
\end{equation}
Also, as described in the ``\emph{Decentralized Token Economy Theory (DeTEcT)}'' \cite{DeTEcT}, we define interactions between agent categories $A, A^{\prime}\in E_{t}$ through the parameters $\beta_{AA^{\prime}}$, known as the interaction rates, while the rotations are defined by $\gamma_{AA^{\prime}}$, which are called the rotation rates. Putting all of these ingredients together, we obtain the generalized form of the dynamical system with constant parameters that models the dynamics of wealth distribution,
\begin{equation}
\label{eq:RelatedWork:DSGeneralizedConstantParams}
	\frac{\Delta}{\Delta t}[\vv{F}(t)] = \frac{1}{M}\vv{F}(t)\odot[\mathcal{B}\cdot\vv{F}(t)]+\Gamma\cdot\vv{F}(t) \quad t\in T,
\end{equation}
where $\cdot$ is the matrix-vector product, $\odot$ is point-wise vector multiplication,
\begin{equation}
	(A\odot B)_{ij} = A_{i}B_{ij},
\end{equation}
$\vv{F}(t)$ is the vector of wealth functions at time $t$,
\begin{equation}
	\vv{F}(t) = (F(A_{1},t), ..., F(A_{n},t))^{T}, \quad A_{1},...,A_{n}\in E_{t},
\end{equation}
$\mathcal{B}$ is an antisymmetric matrix of interaction rates $\beta$,
\begin{equation}
	\mathcal{B} =
	\begin{bmatrix}
		0 & \beta_{A_{1}A_{2}} & \hdots & \beta_{A_{1}A_{n}} \\
		-\beta_{A_{1}A_{2}} & 0 & \hdots & \beta_{A_{2}A_{n}} \\
		\vdots & \vdots & \ddots & \vdots \\
		-\beta_{A_{1}A_{n}} & -\beta_{A_{2}A_{n}} & \hdots & 0\\
	\end{bmatrix},
\end{equation}
and $\Gamma$ is the matrix of rotation rates where each column sums up to zero,
\begin{equation}
	\Gamma = 
	\begin{bmatrix}
		-\gamma_{A_{1}} & \gamma_{A_{2}A_{1}} & \hdots & \gamma_{A_{n}A_{1}} \\
		\gamma_{A_{1}A_{2}} & -\gamma_{A_{2}} & \hdots & \gamma_{A_{n}A_{2}} \\
		\vdots & \vdots & \ddots & \vdots \\
		\gamma_{A_{1}A_{n}} & \gamma_{A_{2}A_{n}} & \hdots & -\gamma_{A_{n}}\\
	\end{bmatrix},
\end{equation}
such that the diagonal elements of $\Gamma$ are
\begin{equation}
	\gamma_{A_{m}} = \sum^{n}_{j\neq m}\gamma_{A_{m}A_{j}}.
\end{equation}
These are the fundamental blocks of our framework that we will be using and generalizing in later sections in this paper.
\subsection{Wealth Distribution Models}
\label{sec:RelatedWork:WealthDistributionModels}
In order to emulate the wealth distribution dynamics in a closed economic system, a family of kinetic theory models has been proposed \cite{KineticTheoryModelsForTheDistributionOfWealth}. These models were originally used to model random interactions between gas molecules in a closed container, but in the context of wealth distribution dynamics they are used to simulate random transactions between the participants (agents) of the economy to find the equilibrium of wealth distribution.\par
The models from this family follow a transaction rule,
\begin{equation}
\label{eq:RelatedWork:TradingRule}
\begin{split}
	& x^{\prime}_{i} = x_{i}-\Delta x, \\
	& x^{\prime}_{j} = x_{j}+\Delta x,
\end{split}
\end{equation}
where $i,j\in\mathbb{N}$ are the labels of the agents, $x_{i}$ and $x_{j}$ are the wealths of the agents respectively before the transaction, $x^{\prime}_{i}$ and $x^{\prime}_{j}$ are their wealths after the transaction, and $\Delta x$ is the transaction quantity. Based on this transaction rule, different models have been introduced where the term $\Delta x$ is used to model different behaviours of agents in an economy.
\subsubsection{Basic Model Without Saving}
The basic model, where agents do not preserve a portion of their wealth, has been proposed by Dragulescu at al. \cite{StatisticalMechanicsOfMoney}, where the trading rule becomes
\begin{equation}
\label{eq:RelatedWork:NoSavingModel}
\begin{split}
	& x^{\prime}_{i} = \epsilon(x_{i}+x_{j}), \\
	& x^{\prime}_{j} = \tilde{\epsilon}(x_{i}+x_{j}),
\end{split}
\end{equation}
with $\epsilon\sim U(0,1)$ and $\tilde{\epsilon}$ its complementary fraction (i.e., $\epsilon+\tilde{\epsilon}=1$). The form of the trading rule from equation \ref{eq:RelatedWork:TradingRule} is recovered by setting
\begin{equation}
	\Delta x = \tilde{\epsilon}x_{i}-\epsilon x_{j}.
\end{equation}
In a closed economic system with random transactions taking place according to the trading rule in equation \ref{eq:RelatedWork:NoSavingModel}, the wealth distribution becomes a Boltzmann distribution and leads to inequitable wealth distribution where a few agents hold majority of the wealth.
\subsubsection{Model With Constant Global Saving Propensity}
In most economic systems, the agents tend to preserve some proportion of their wealth which can be represented by a wealth distribution model through the introduction of a saving propensity parameter $0<\lambda<1$, which represents the proportion of wealth an agent saves, therefore, not using it in the transaction. This model has been introduced by Chakraborti et al. \cite{StatisticalMechanicsOfMoneySavingPropensity}, where the trading rule is defined as
\begin{equation}
\label{eq:RelatedWork:GlobalSavingPropensity}
\begin{split}
	& x^{\prime}_{i} = \lambda x_{i}+\epsilon(1-\lambda)(x_{i}+x_{j}), \\
	& x^{\prime}_{j} = \lambda x_{j}+\tilde{\epsilon}(1-\lambda)(x_{i}+x_{j}),
\end{split}
\end{equation}
with $\epsilon$ and $\tilde{\epsilon}$ defined as before. To recover the general trading rule in equation \ref{eq:RelatedWork:TradingRule}, the reallocated wealth is defined by
\begin{equation}
	\Delta x = (1-\lambda)[\tilde{\epsilon}x_{i}-\epsilon x_{j}].
\end{equation}
Note that $\lambda$ is the same regardless of the agent who undergoes transaction. Therefore, this model implicitly assumes that every agent has the same preference for how much wealth they will save before a transaction. The wealth distribution after simulating the random transactions with the rule in equation \ref{eq:RelatedWork:GlobalSavingPropensity} results in the Gamma distribution of wealth between the agents.
\subsubsection{Model With Individual Saving Propensities}
Individual saving propensities $\{\lambda_{i}:0<\lambda_{i}<1\}$ for agent $i$ can be introduced to add the individual preference of the agents to save a specific portion of their wealth before transacting. The trading rule with individual saving propensities proposed by Chatterjee et al. \cite{MoneyInGasLikeMarkets} is
\begin{equation}
\label{eq:RelatedWork:IndividualSavingPropensity}
\begin{split}
	& x^{\prime}_{i} = \lambda_{i}x_{i}+\epsilon[(1-\lambda_{i})x_{i}+(1-\lambda_{j})x_{j}], \\
	& x^{\prime}_{j} = \lambda_{j}x_{j}+\tilde{\epsilon}[(1-\lambda_{i})x_{i}+(1-\lambda_{j})x_{j}],
\end{split}
\end{equation}
where the general trading rule in equation \ref{eq:RelatedWork:TradingRule} is obtained by setting
\begin{equation}
	\Delta t = \tilde{\epsilon}(1-\lambda_{i})x_{i}-\epsilon(1-\lambda_{j})x_{j}.
\end{equation}
If multiple simulations are run using the trading rule in equation \ref{eq:RelatedWork:IndividualSavingPropensity} with different individual saving propensity settings each time, the average of the equilibrium wealth distributions is the Pareto distribution. Pareto distribution is commonly used in economics to model the wealth distribution in a society as it describes the disparity of wealth between the ``wealthy'' and the ``poor'' agents in the economy.
\subsubsection{Summary}
The wealth distribution models outlined here are well-established in the academic literature, and are commonly used for modelling generic macroeconomic principles. For example, the relaxation time for a wealth distribution to return to its equilibrium can be measured \cite{RelaxationInStatisticalManyAgentEconomyModels}, or in the case of the model with individual saving propensities, it is used to describe the Pareto principle \cite{ParetoPrinciple} that is observed in real-world economies.

\section{Parametrization Extension and Derivation of Existing Wealth Distribution Models}
\label{sec:DerivingWealthModels}
\subsection{Parameter Modification Techniques}
In our framework, one of the assumptions we used is that the parameters (i.e., interaction and rotation rates) are constant. This is useful when we either simulate wealth distribution based on some ``predefined trends'', or when we are trying to find the parameters of the dynamical system that allow the dynamical system to reach a desired wealth distribution (i.e., a desired attractor). However, if we drop this assumption and allow parameters to change at different time steps, we will see that DeTEcT becomes a broad framework that can be parametrized in a certain way to replicate existing wealth distribution models.\par
However, before obtaining other wealth distribution models from DeTEcT, we would like to introduce a taxonomy of different parametrizations that we can use to modify our original approach \cite{DeTEcT}.\par
We broadly categorize interaction and rotation rates using two features, namely whether the parameters are static or dynamic, and whether the parameters are probabilistic or deterministic. These two features allow us to cover all possible extensions that can be introduced to the DeTEcT framework. Table \ref{table:DerivingWealthModels:ParamsTaxonomy} demonstrates some of the applications of different parameter modifications, which we will examine on a case-by-case basis.\par
\begin{table}[h!]
\centering
\begin{tabular}{ |c|c|c| } 
\hline
& \textbf{Deterministic} & \textbf{Probabilistic} \\
\hline
\textbf{Static} & \begin{tabular}{@{}c@{}}Implied parameters for simulation, \\ Obtained from inverse propagation \end{tabular} & \begin{tabular}{@{}c@{}}Averaged parameters from \\ historical data \end{tabular} \\
\hline
\textbf{Dynamic} & \begin{tabular}{@{}c@{}}Changing parameters \\ based on incoming data\end{tabular} & \begin{tabular}{@{}c@{}}Stochastically simulated \\ parameters \end{tabular} \\
\hline
\end{tabular}
\caption{Applications of Parameter Modification Techniques}
\label{table:DerivingWealthModels:ParamsTaxonomy}
\end{table}
Furthermore, we can categorize dynamic parameters by whether they are proactive or reactive. By proactive dynamic parameters we mean that these parameters are either set in advance for every time step in the simulation, or are derived from a arbitrary function that does not consider any history of economic activity (i.e., the function is agnostic to the state of the economy). Reactive dynamic parameters are the parameters that are obtained from some arbitrary function that takes the state of the economy as an input to produce parameters for the next time step (e.g., a change in the sizes of incentives paid out to all agent categories follows the review of an inflation reading collected in the last time step).\par
\begin{table}[h!]
\centering
\begin{tabular}{ |c|c|c| } 
\hline
& \textbf{Deterministic} & \textbf{Probabilistic} \\
\hline
\textbf{Proactive} & \begin{tabular}{@{}c@{}}Predefined program of \\ changing parameters \end{tabular} & \begin{tabular}{@{}c@{}}Predefined program \\ of changing stochastic parameters\end{tabular} \\
\hline
\textbf{Reactive} & \begin{tabular}{@{}c@{}c@{}}Parameters generated \\ in response \\ to the state of economy\end{tabular} & \begin{tabular}{@{}c@{}c@{}}Stochastic parameters \\ generated \\ in response to the state of economy \end{tabular} \\
\hline
\end{tabular}
\caption{Applications of Dynamic Parameter Modification Techniques}
\label{table:DerivingWealthModels:DynamicParamsTaxonomy}
\end{table}
Table \ref{table:DerivingWealthModels:DynamicParamsTaxonomy} describes use cases for proactive and reactive dynamic parameters. Whether dynamic parameters are proactive or reactive depends on how $\mathcal{B}(t)$ and $\Gamma(t)$ are defined, but this categorization doesn't change the general form of the dynamical system for the case of the dynamic parameters. We make this distinction between dynamic parameters to demonstrate that there can be a more detailed taxonomy of the parametrization methods, but we will not expand on this further in this paper.\par
Since we will be looking at modifying the parameters of the dynamical system, we reiterate the original definition of interaction rates, that we have used before \cite{DeTEcT},
\begin{equation}
	\beta_{AA^{\prime}} = \frac{M}{\Delta t}\frac{\sum_{i_{AA^{\prime}}\in I_{AA^{\prime}}}\iota(i_{AA^{\prime}},t)}{F(A,t)F(A^{\prime},t)} \quad \forall t\in T,
\end{equation}
where $\Delta t$ is the difference between two consequent time steps, $i_{AA^{\prime}}$ are is an interaction type defined between agent categories $A,A^{\prime}\in E_{t}$, $I_{AA^{\prime}}$ is the set of all possible interaction types defined in the tokenomic taxonomy between $A$ and $A^{\prime}$, and $\iota(i_{AA^{\prime}},t)$ is all wealth redistributed between $A$ and $A^{\prime}$ due to interactions of type $i_{AA^{\prime}}$ that took place in the interval $(t-\Delta t, t)$.\par
Essentially, we can think of $\beta_{AA^{\prime}}$ as the net wealth redistributed between $A$ and $A^{\prime}$ expressed as the proportion of wealth that these agent categories have. Using the analogy from section \ref{sec:RelatedWork:DeTEcT}, let $k_{AA^{\prime}}$ be a purchase of a car, and $j_{AA^{\prime}}$ be a purchase of a truck - both of these are possible interaction types that can take place between \emph{Car Manufacturer} and \emph{Household}, while $\iota$ is the wealth that is reallocated from \emph{Household} to \emph{Car Manufacturer} for purchasing these goods, such that
\begin{equation}
\label{eq:DerivingWealthModels:IotaDefinition}
	D_{i_{AA^{\prime}}}(t)P_{i_{AA^{\prime}}}(t)\coloneqq\iota(i_{AA^{\prime}},t),
\end{equation}
where $D_{i_{AA^{\prime}}}(t)$ is the demand for interaction of type $i_{AA^{\prime}}$ (i.e., how many cars or trucks were purchased since the last time step), and $P_{i_{AA^{\prime}}}(t)$ is the price of each interaction (i.e., the price at which these goods were purchased).\par
We can now explore how by modifying the definitions above we can make DeTEcT more flexible.
\subsubsection{Static Deterministic Parameters}
Static deterministic parameters are the ones that we have defined in the section above. These are the basic building blocks upon which we can implement our modifications.\par
The fundamental assumption we use here is that the parameter values do not change in time, and in order for that to happen we must balance net wealth redistributed (i.e., $\iota$) and the wealth functions of the agent categories (i.e., $F(A,t)$ and $F(A^{\prime},t)$). We assume that at every time step the prices $P_{i_{AA^{\prime}}}(t)$ are readjusted against the given demands $D_{i_{AA^{\prime}}}(t)$ and wealth functions for that time period.These restrictions give rise to the following effects:
\begin{itemize}
	\item \textbf{Forward Propagation}: Inaccuracies in modelling real economies, as in real economies interaction and rotation rates will be changing in time. The parameters themselves are derived from arbitrary historical transaction data, which limits the possible cases. However, the numerical simulation with these ``simple'' parameters is very quick and straight forward.
	\item \textbf{Inverse Propagation}: If economy governors have the control over the price setting mechanisms, maintaining constant parameters is easy as it is done via setting appropriate prices. However, in the case of a free market economy, the parameters would not be maintained at a constant level as market participants will be free to set their own prices in response to changing demand. The numerical solution employed to perform inverse propagation must be run only once to obtain the desired values of interaction and rotation rates.
\end{itemize}
The interaction rates are defined as
\begin{equation}
	\beta_{AA^{\prime}} = \frac{M}{\Delta t}\frac{\sum_{i_{AA^{\prime}}\in I_{AA^{\prime}}}\iota(i_{AA^{\prime}},t)}{F(A,t)F(A^{\prime},t)} \quad \forall t\in T,
\end{equation}
while the rotation rates are just constants
\begin{equation}
	\gamma_{AA^{\prime}} = c, \quad c\in\mathbb{R}.
\end{equation}
\subsubsection{Static Probabilistic Parameters}
Static probabilistic parameters are an extension of the static deterministic parameters, where we add a term that represents a probability of success for each interaction type $i_{AA^{\prime}}$ and it tells us how many interactions of this type we expect at each time step. The effects of this parametrization technique are not much different to those of static deterministic parameters. The notable difference is that now we can use historical data to create a probability distribution for the demands $D_{i_{AA^{\prime}}}(t)$ and at each time step we aim to rebalance these ``expected'' demands with the prices.\par
The benefit of using static probabilistic parameters as opposed to static deterministic ones, is that we simulate the wealth distribution based on the historical data compared to some arbitrary demands. We can use any probability distribution to simulate the distribution of demands, but the only restriction we impose for now is that the choice of probability distribution should match the data format of the demand (i.e., demand cannot be $1.5$ if the interaction is \emph{Buying car}, but it can be $1.5kg$ if the interaction is \emph{Buying rice} for $\$100\textit{ per kg}$).\par
For example, assume that for each interaction type $i_{AA^{\prime}}\in I_{AA^{\prime}}$ there exists a probability $p_{i_{AA^{\prime}}}\in[0,1]$ that the interaction of this type takes place with the Bernoulli distribution $\mathcal{I}_{i_{AA^{\prime}}}\sim Bernoulli(p_{i_{AA^{\prime}}})$. Now we can write interaction rates as
\begin{equation}
	\beta_{AA^{\prime}} = \frac{M}{\Delta t}\frac{\sum_{i_{AA^{\prime}}\in I_{AA^{\prime}}}E[\mathcal{I}_{i_{AA^{\prime}}}]\iota(i_{AA^{\prime}}, t)}{F(A,t)F(A^{\prime}, t)} = \frac{M}{\Delta t}\frac{\sum_{i_{AA^{\prime}}\in I_{AA^{\prime}}}p_{i_{AA^{\prime}}}\iota(i_{AA^{\prime}}, t)}{F(A,t)F(A^{\prime}, t)}, \quad t\in T,
\end{equation}
where $p_{i_{AA^{\prime}}}$ acts as a weight for the summation over $\iota$. If we plug equation \ref{eq:DerivingWealthModels:IotaDefinition} into equation above, we get
\begin{equation}
	\beta_{AA^{\prime}} = \frac{M}{\Delta t}\frac{\sum_{i_{AA^{\prime}}\in I_{AA^{\prime}}}p_{i_{AA^{\prime}}}D_{i_{AA^{\prime}}}(t)P_{i_{AA^{\prime}}}(t)}{F(A,t)F(A^{\prime}, t)} = \frac{M}{\Delta t}\frac{\sum_{i_{AA^{\prime}}\in I_{AA^{\prime}}}E[\tilde{\mathcal{I}}_{i_{AA^{\prime}}}]P_{i_{AA^{\prime}}}(t)}{F(A,t)F(A^{\prime}, t)},
\end{equation}
where $\tilde{\mathcal{I}}_{i_{AA^{\prime}}}\sim B(D_{i_{AA^{\prime}}}(t), p_{i_{AA^{\prime}}})$ is the binomial distribution for multiple interactions of type $i_{AA^{\prime}}$.\par
We can generalize this notion by defining the static probabilistic interaction rates as
\begin{equation}
\label{eq:DerivingWealthModels:StaticProbabilisticParam}
	\beta_{AA^{\prime}} = \frac{M}{\Delta t}\frac{\sum_{i_{AA^{\prime}}\in I_{AA^{\prime}}}E[\tilde{\mathcal{D}}_{i_{AA^{\prime}}}]P_{i_{AA^{\prime}}}(t)}{F(A,t)F(A^{\prime}, t)},
\end{equation}
where $\tilde{\mathcal{D}}_{i_{AA^{\prime}}}$ is the probability distribution of demand for interaction type $i_{AA^{\prime}}$, assuming that the distribution of demands is well-defined (i.e., no fractions of car being sold).\par
The rotation rates $\gamma_{AA^{\prime}}$ can also be distributed according to a probability distribution that is based on the historical data. Recall that rotations describe the reallocation of wealth from one agent category to another due to some individual agents changing their agent categories between time steps. Given that we have a well-defined tokenomic taxonomy, we can look at historical data to see what agents have changed their agent categories and, therefore, we can estimate on average how much wealth is rotated between different agent categories per time step. In this case, the rotation rates are defined as
\begin{equation}
	\gamma_{AA^{\prime}} = \mu_{AA^{\prime}}, \quad \mu_{AA^{\prime}}\in\mathbb{R},
\end{equation}
where $\mu_{AA^{\prime}}$ is the average gross wealth moved from agent category $A$ to $A^{\prime}$ due to rotations of agents from $A$ to $A^{\prime}$.
\subsubsection{Dynamic Deterministic Parameters}
Now, we can add a further generalization to our framework by defining the dynamic deterministic parameters. The dynamic deterministic parameters are the interaction and rotation rates that are time-dependent (i.e., $\beta_{AA^{\prime}}(t)$ and $\gamma_{AA^{\prime}}(t)$). The motivation for introducing time-dependence is that when we work with real-world data it is likely that the wealth redistribution rates will vary in time, in particular as an economy goes through economic cycles. When using the static parameters we assume that economies balance out demands and prices to construct price hyperplanes, but in the case of real-world economies this is generally not the case.\par
The effect of introducing the dynamic deterministic parameters on the simulation has major impacts:
\begin{itemize}
	\item \textbf{Forward Propagation}: Since the parameters can be arbitrary, we could fit our dynamical system on any dynamics with the same number of agent categories and, therefore, fit our model onto any economy. This might be useful when we are changing parameters based on the incoming data to perform analysis of economic activity in an economy, but this is not helpful when we are trying to simulate future results as the dynamic parameters are too flexible.
	\item \textbf{Backward Propagation}: With the dynamic parameters we can perfectly fit our model on the empirical data as at every time step we can change the rates appropriately to link wealth distributions at different time steps. However, this is overfitting of the model, and the rates we obtain this way are not useful for making any predictions. Overfitting can be prevented using different techniques such as adding penalty terms to the loss function of the numerical solver, bootstrapping results of multiple simulations with random initial positions, or even splitting the data into compartments (e.g., split economy data by the timestamps of policy implementations) and training the model on individual compartments. We will leave the discussion of simulation engine design, model fitting, and methods used for prevention of overfitting to our future paper.
\end{itemize}
In order for the interaction  rates to become dynamic we must lift the constraint that we imposed when defining static deterministic interaction rates - that the vector of demands and the vector of prices will balance out the wealth functions of agent categories $A$ and $A^{\prime}$. Dropping this assumption results in the dynamic interaction rates defined as
\begin{equation}
\label{eq:DerivingWealthModels:DynamicDeterministicParam}
	\beta_{AA^{\prime}}(t) = \frac{M}{\Delta t}\frac{\sum_{i_{AA^{\prime}}\in I_{AA^{\prime}}}\iota(i_{AA^{\prime}},t)}{F(A,t)F(A^{\prime},t)} \quad \forall t\in T.
\end{equation}
The rotation rates become dynamic by assuming that for a rotation of wealth from $A$ to $A^{\prime}$ there exists a discrete function $g_{AA^{\prime}}:T\rightarrow\mathbb{R}$ well-defined $\forall t\in T$ such that
\begin{equation}
	\gamma_{AA^{\prime}}(t) = g_{AA^{\prime}}(t).
\end{equation}
Note that since we removed the rebalancing constraint, in the forward propagation we are free to change the function $\iota(i_{AA^{\prime}},t)$ inside our simulation. This means that both $D_{i_{AA^{\prime}}}(t)$ and $P_{i_{AA^{\prime}}}(t)$ can be arbitrary deterministic functions and, therefore, we can choose any deterministic models to simulate these in the forward propagation with the dynamic deterministic parameters extension.
\subsubsection{Dynamic Probabilistic Parameters}
The last case of parameter modification is the dynamic probabilistic parameters, where interaction and rotation rates are time-dependent and are defined by stochastic processes. The motivation for this procedure is that we can create stochastic processes that satisfy the distribution of empirical data at different time steps, and we can feed these processes inside the simulation in order to gauge what an economy may look like.\par
Compared to the dynamic deterministic parameters, we now have explainable parameters that are based on the distribution of historical data. This allows for a good risk assessment tool, where the economic dynamics is simulated multiple times, so that we can see what are the possible ``bad'' cases for the economy, how they may play out, and what are the impacts of these cases on different agent categories. However, this means that the backward propagation is fruitless in this case as the stochastic processes will ensure that every time we run the dynamical system the dynamic of the simulated wealth distribution will be different.\par
The static probabilistic parameters were defined using the probability distribution $\tilde{\mathcal{D}}_{i_{AA^{\prime}}}$ that simulates the demands at each time step, and then the interaction rates are kept constant by rebalancing the varying demands with prices and wealth functions. We again lift the rebalancing constraint, so that the interaction rates become
\begin{equation}
	\beta_{AA^{\prime}}(t) = \frac{M}{\Delta t}\frac{\sum_{i_{AA^{\prime}}\in I_{AA^{\prime}}}\tilde{\mathcal{D}}_{i_{AA^{\prime}},t}P_{i_{AA^{\prime}}}(t)}{F(A,t)F(A^{\prime},t)},
\end{equation}
where $\{\tilde{\mathcal{D}}_{i_{AA^{\prime}},t}\}_{t\in T}$ is the stochastic process that simulates the demand for the interactions of interaction type $i_{AA^{\prime}}$, and $\tilde{\mathcal{D}}_{i_{AA^{\prime}},t}$ is the random variable that follows this stochastic process.\par
However, lifting the rebalancing constraint also allows for price to become a stochastic process $\{\tilde{\mathcal{P}}_{i_{AA^{\prime}},t}\}_{t\in T}$ for some interaction type $i_{AA^{\prime}}$, with $\tilde{\mathcal{P}}_{i_{AA^{\prime}},t}$ being the random varibale following this process. We classify the interaction rates to be dynamic and probabilistic if either the price, demand, or both are described by stochastic processes. Therefore, the following interaction rates are also defined as dynamic probabilistic parameters:
\begin{equation}
	\beta_{AA^{\prime}}(t) = \frac{M}{\Delta t}\frac{\sum_{i_{AA^{\prime}}\in I_{AA^{\prime}}}D_{i_{AA^{\prime}}}(t)\tilde{\mathcal{P}}_{i_{AA^{\prime}},t}}{F(A,t)F(A^{\prime},t)},
\end{equation}
where demand is a deterministic function and the price is described by the stochastic process, and
\begin{equation}
	\beta_{AA^{\prime}}(t) = \frac{M}{\Delta t}\frac{\sum_{i_{AA^{\prime}}\in I_{AA^{\prime}}}\tilde{\mathcal{D}}_{i_{AA^{\prime}},t}\tilde{\mathcal{P}}_{i_{AA^{\prime}},t}}{F(A,t)F(A^{\prime},t)},
\end{equation}
where both, demand and price are defined by different stochastic processes. Note that is applicable only to forward propagation, as the backward propagation will not work with the dynamic probabilistic parameters.\par
Since all the stochastic processes used are not necessarily the same, all of them have to be simulated and each must be $\lvert T\rvert$ steps long, which are then used inside of the dynamical system to simulate the wealth distribution in an economy.\par
Just like with static probabilistic interaction rates, we have to make sure that the values simulated by the stochastic processes ``make sense'' so that we don't get infeasible demands (e.g., half a car being sold).\par
The rotation rates themselves become stochastic processes and are defined as
\begin{equation}
	\gamma_{AA^{\prime}}(t) = G_{AA^{\prime},t},
\end{equation}
where $\{G_{AA^{\prime},t}\}$ is the stochastic process that simulates the gross wealth redistributed from agent category $A$ to $A^{\prime}$, and $G_{AA^{\prime},t}$ is the random variable associated with it.
\subsection{Relationship to Existing Wealth Distribution Models}
In this section we demonstrate how our framework is related to the existing wealth distribution models. To proceed we must first state that it is common for the wealth distribution to be studied from the perspective of agent-based models where individual agents transact with one another resulting in wealth redistribution.\par
As we have seen in the previous section, we can redefine the dynamical system so that it has dynamic parameters,
\begin{equation}
\label{eq:DerivingWealthModels:GeneralEquation}
	\frac{\Delta}{\Delta t}[\vv{F}(t)] = \frac{1}{M}\vv{F}(t)\odot[\mathcal{B}(t)\cdot\vv{F}(t)]+\Gamma(t)\cdot\vv{F}(t),
\end{equation}
where the matrix of interaction rates $\mathcal{B}(t)$ is
\begin{equation}
	\mathcal{B}(t) =
	\begin{bmatrix}
		0 & \beta_{A_{1}A_{2}}(t) & \hdots & \beta_{A_{1}A_{n}}(t) \\
		\beta_{A_{2}A_{1}}(t) & 0 & \hdots & \beta_{A_{2}A_{n}}(t) \\
		\vdots & \vdots & \ddots & \vdots \\
		\beta_{A_{n}A_{1}}(t) & \beta_{A_{n}A_{2}}(t) & \hdots & 0\\
	\end{bmatrix},
\end{equation}
and the matrix of rotation rates $\Gamma(t)$ is
\begin{equation}
	\Gamma(t) =
	\begin{bmatrix}
		-\gamma_{A_{1}}(t) & \gamma_{A_{2}A_{1}}(t) & \hdots & \gamma_{A_{n}A_{1}}(t) \\
		\gamma_{A_{1}A_{2}}(t) & -\gamma_{A_{2}}(t) & \hdots & \beta_{A_{n}A_{2}}(t) \\
		\vdots & \vdots & \ddots & \vdots \\
		\gamma_{A_{1}A_{n}}(t) & \beta_{A_{2}A_{n}}(t) & \hdots & -\gamma_{A_{n}}(t)\\
	\end{bmatrix}.
\end{equation}
This new generalized definition of the dynamical system can help us with the task of connecting DeTEcT to existing wealth distribution models, as they tend to use dynamic parameters.\par
Also, in the context of our framework, let us assume that the agent categories contain only one agent, such that each individual agent has his own agent type (i.e., $\forall a\in\Lambda!\exists A\in E_{t}$ such that $\{a\}=A$, and in general $\lvert E_{t}\rvert=\lvert\Lambda\rvert$). Having unique agent categories for every individual agent also implies that rotations are not defined in this scenario as agents cannot change agent categories, so the matrix of rotation rates is a zero matrix (i.e., $\Gamma(t)=0_{n\times n}$).\par
\begin{remark}
	The assumption that there is only one agent per agent category is not necessary for the wealth distribution models we are about to describe; what's necessary is the assumption that there are no rotations defined between the agent categories. The reason we assume that there is only one agent per agent category is for the demonstration of how our framework integrates with the existing wealth distribution models, where the notion of agent category is not defined. In this section we refer to $A_{j}$ as an individual agent, since it is a set $A_{j}=\{a_{j}\}$ with the cardinality $\lvert A_{j}\rvert=1$.
\end{remark}
The notions of interaction types and interaction quantities are also exclusive to DeTEcT; in previous literature on wealth distribution models, the net wealth redistributed between agents is directly defined by one transaction at time step $t$. Without the loss of generality, let the net wealth redistributed be defined as
\begin{equation}
	\Delta F_{AA^{\prime}}(t)\coloneqq\sum_{i_{AA^{\prime}}\in I_{AA^{\prime}}}\iota(i_{AA^{\prime}},t),
\end{equation}
where $A$ and $A^{\prime}$ are agent categories that contain one agent each. $\Delta F_{AA^{\prime}}(t)$ is the net wealth that has been transacted between individual agents $a\in A$ and $a^{\prime}\in A^{\prime}$ at time step $t\in T$.\par
With these limitations, we can rewrite the equation \ref{eq:DerivingWealthModels:GeneralEquation} so that the dynamical system becomes
\begin{equation}
\label{eq:DerivingWealthModels:ExistingModelsEquation}
	\frac{\Delta}{\Delta t}[\vv{F}(t)] = \frac{1}{M}\vv{F}(t)\odot[\mathcal{B}(t)\cdot\vv{F}(t)],
\end{equation}
with the individual interaction rates now being given by
\begin{equation}
\label{eq:DerivingWealthModels:ExistingModelsBeta}
	\beta_{AA^{\prime}}(t) = \frac{M}{\Delta t}\frac{\Delta F_{AA^{\prime}}(t)}{F(A,t)F(A^{\prime},t)}.
\end{equation}
We can further simplify the equation \ref{eq:DerivingWealthModels:ExistingModelsEquation}, by applying the finite difference method. If there are $n$ agents in the economy, we have that $\lvert\Lambda\rvert=n=\lvert E_{t}\rvert$ implying that the vector $\vv{F}(t)$ is an $n$-dimensional vector of wealth functions for every agent. Switching from vector notation to index notation and applying the finite difference method, we obtain a system with $n$ dynamical equations
\begin{equation}
	F(A_{j},t) = F(A_{j},t-\Delta t)+\frac{\Delta t}{M}[F(A_{j},t-\Delta t)\sum^{n}_{k\neq j}\beta_{A_{j}A_{k}}(t-\Delta t)F(A_{k},t-\Delta t)].
\end{equation}
Using equation \ref{eq:DerivingWealthModels:ExistingModelsBeta} we simplify the equation above to
\begin{equation}
\label{eq:DerivingWealthModels:GeneralizedTransactionRule}
	F(A_{j},t) = F(A_{j},t-\Delta t)+\sum^{n}_{k\neq j}\Delta F_{A_{j}A_{k}}(t-\Delta t),
\end{equation}
which is a version of the dynamical system that is often considered in the literature and studies on wealth distribution \cite{KineticTheoryModelsForTheDistributionOfWealth}.\par
Equation \ref{eq:DerivingWealthModels:GeneralizedTransactionRule} describes the redistribution of wealth between agents in the economy from the perspective of individual transactions, and is referred to as the \emph{transaction rule}. The terms $\Delta F_{A_{j}A_{k}}(t-\Delta t)$ are the net transactions that take place between agents such that the invariance of wealth assumption is satisfied,
\begin{equation}
	M = \sum_{A\in E_{t}}F(A,t).
\end{equation}
Without the loss of generality, we can pick a time interval $\Delta t$ such that an individual agent performs at most only one transaction in this interval. This further reduces the problem, where the transaction rule is simplified to become
\begin{equation}
\label{eq:DerivingWealthModels:TransactionRule}
	F(A_{j},t) = F(A_{j},t-\Delta t)+\Delta F_{A_{j}A_{k}}(t-\Delta t),
\end{equation}
and the change in wealth of the counterparty $A_{k}$, is defined as
\begin{equation}
\label{eq:DerivingWealthModels:TransactionRuleCounterparty}
	F(A_{k},t) = F(A_{k},t-\Delta t)+\Delta F(A_{k}A_{j})(t-\Delta t).
\end{equation}
Note that for the invariance of wealth to be satisfied, the net change in wealth is antisymmetric,
\begin{equation}
	\Delta F_{A_{j}A_{k}}(t-\Delta t) = -\Delta F_{A_{k}A_{j}}(t-\Delta t).
\end{equation}
With the reformulated transaction rule defined in by equations \ref{eq:DerivingWealthModels:TransactionRule} and \ref{eq:DerivingWealthModels:TransactionRuleCounterparty}, along with the antisymmetric constraint on net change in wealth, we can demonstrate how existing wealth distribution models are derived from our framework.
\subsubsection{Wealth Distribution Model with No Saving}
First model we consider is the wealth distribution model with no saving, introduced by Dragulescu et al. \cite{StatisticalMechanicsOfMoney}. This model assumes that an agent transacts a random portion of his wealth at time step $t$.\par
Given the transaction rule defined in equations \ref{eq:DerivingWealthModels:TransactionRule} and \ref{eq:DerivingWealthModels:TransactionRuleCounterparty}, the net change in wealth in this model is defined to be
\begin{equation}
	\Delta F(A_{j}A_{k})(t-\Delta t) = \bar{\epsilon}_{A_{j}A_{k}}F(A_{k},t-\Delta t)-\epsilon_{A_{j}A_{k}}F(A_{j},t-\Delta t),
\end{equation}
where $\epsilon\sim U(0,1)$ is a random proportion of the combined wealth of $A_{j}$ and $A_{k}$, that the agent $A_{j}$ will have after the transaction. Given this definition of net change of wealth, the transaction rule can be rewritten as
\begin{equation}
\label{eq:DerivingWealthModels:TransactionRuleBoltzmann}
\begin{split}
	& F(A_{j},t) = \epsilon_{A_{j}A_{k}}(F(A_{j},t-\Delta t)+F(A_{k},t-\Delta t)) \\
	& F(A_{k},t) = \bar{\epsilon}_{A_{j}A_{k}}(F(A_{j},t-\Delta t)+F(A_{k},t-\Delta t)),
\end{split}
\end{equation}
where $\bar{\epsilon}_{A_{j}A_{k}}$ is the complimentary fraction of $\epsilon_{A_{j}A_{k}}$ (i.e., $\epsilon_{A_{j}A_{k}}+\bar{\epsilon}_{A_{j}A_{k}}=1$).\par
This model describes the wealth distribution dynamics through random interactions, similar to a kinetic model of gas particles interacting inside a closed container. The simulation of this model produces an equilibrium of wealth distribution (i.e., the attractor of the dynamical system describing the wealth distribution) and as shown by Dragulescu et al. \cite{StatisticalMechanicsOfMoney} it fits to the Boltzmann distribution (else known as the Gibbs distribution)
\begin{equation}
\label{eq:DerivingWealthModels:NoSavingBoltzmann}
	g(A_{j}) = \frac{1}{\langle F\rangle}e^{-\frac{F(A_{j})}{\langle F\rangle}},
\end{equation}
with
\begin{equation}
	\langle F\rangle = \frac{M}{n},
\end{equation}
where $M$ is the maximum supply, and $n$ is the number of agents in the economy (i.e., $n=\lvert\Lambda\rvert$). In thermodynamic terms, the \emph{average temperature} of the system is the average wealth $\langle F\rangle$ of an agent in the economy \cite{KineticTheoryModelsForTheDistributionOfWealth}, and the \emph{energy of the state} is the wealth $F(A_{j})$ of an agent $A_{j}$. Note that we dropped the time-dependence of the wealth function $F$ given that the distribution of wealth between agents reaches an equilibrium and since the system attained its attractor the change of wealth over time is negligible. The equation \ref{eq:DerivingWealthModels:NoSavingBoltzmann} is the probability density function that defines the probability that an agent has wealth $F(A_{j})$.\par
An interesting feature of the Boltzmann distribution in equation \ref{eq:DerivingWealthModels:NoSavingBoltzmann} is its ``robustness'' with respect to different factors such as initial conditions or multi-agent interactions, which do not impact the accuracy of the fit of the Boltzmann distribution over the equilibrium of the wealth distribution obtained from the trading rule in the equation \ref{eq:DerivingWealthModels:TransactionRuleBoltzmann}.\par
In terms of economic modelling this result implies that if the transactions in the economy are random and are more akin to gambling, then the majority of wealth will belong to a small number of individuals and most agents in the economy will be poor. Also, due to the robustness of Boltzmann distribution, we can expect that no matter the initial state, the economy where agents act according to the transaction rule \ref{eq:DerivingWealthModels:TransactionRuleBoltzmann} will always result in having small number of ``rich'' and large number of ``poor'' agents.\par
The reason this model is said to have ``no saving'' is because $\epsilon_{A_{j}A_{k}}$ can be 0 or 1, which leads to an agent obtaining all the wealth from the counterparty. In this context, \emph{saving} is defined as the proportion of wealth that an agent is guaranteed not to transact regardless of the value of the parameter $\epsilon_{A_{j}A_{k}}$. We will consider the implementation of the ``saving parameter'' (i.e., \emph{saving propensity}) later in this section.\par
\subsubsection{Wealth Distribution Model with Minimum Investment and No Saving}
The model proposed by Chakraborti \cite{DistributionsOfMoneyInModelMarketsOfEconomy} introduces a concept of minimum transaction value, $F_{min}$, as a means to simulate a system where agents ``invest'' the same wealth with the outcome of the investment being the wealth each agent receives from this joint investment pool. As in the previous model, the outcome of the investment (i.e., transaction) is random and agents do not save any proportion of their wealth, risking it instead in a form of an investment.\par
The net change in wealth is defined as
\begin{equation}
	\Delta F_{A_{j}A_{k}}(t-\Delta t) = (2\epsilon_{A_{j}A_{k}}-1)F_{min},
\end{equation}
where $\epsilon_{A_{j}A_{k}}\sim U(0,1)$ as before, and 
\begin{equation}
	F_{min}=min(F(A_{j},t-\Delta t),F(A_{k},t-\Delta t))
\end{equation}
is the minimum investment that both agents $A_{j}$ and $A_{k}$ make, which is equivalent to the wealth of the agent with the smaller wealth (e.g., if $F(A_{j},t-\Delta t)<F(A_{k},t-\Delta t)$ then the minimum investment is $F_{min}=F(A_{j},t-\Delta t)$).\par
Applying this definition of net change in wealth to the general transaction rule (equations \ref{eq:DerivingWealthModels:TransactionRule} and \ref{eq:DerivingWealthModels:TransactionRuleCounterparty}), we get
\begin{equation}
\begin{split}
	& F(A_{j},t) = F(A_{j},t-\Delta t)+(2\epsilon_{A_{j}A_{k}}-1)\times min(F(A_{j},t-\Delta t),F(A_{k},t-\Delta t)) \\
	& F(A_{k},t) = F(A_{k},t-\Delta t)+(2\epsilon_{A_{j}A_{k}}-1)\times min(F(A_{j},t-\Delta t),F(A_{k},t-\Delta t)).
\end{split}
\end{equation}
The dynamics of this system is unique in the sense that over the course of the simulation, agents will loose their wealth such that $F_{min}=0$, which means they cannot invest any more or participate in the transactions. This implies that the agents are being ``driven out of the market'' once they run out of wealth to transact.\par
This peculiarity of the model leads to the wealth distribution equilibrium to be described by a power-law (i.e., $g(A_{j}, t-\Delta t)\sim F(A_{j},t-\Delta t)^{-v}$ with some exponent parameter $v$ for a given time step $t-\Delta t$) with an exponentially falling tail (i.e., $g(A_{j},t-\Delta t)\sim e^{-\alpha F(A_{j},t-\Delta t)}$ with some parameter $\alpha$ for a given time step $t-\Delta t$), and as $t\rightarrow\infty$ all agents apart from one are driven out of the market, while that one agent holding all maximum supply $M$ (i.e., $\lim_{t\rightarrow\infty}F(A_{j}, t)=M$). From the numerical simulations, it is estimated that for $t=15,000,000$ more than 99\% of agents are driven out of the market \cite{DistributionsOfMoneyInModelMarketsOfEconomy}.\par
Note that here we did not drop the time-dependence of the wealth function $F$ since the probability density function is defined at the time steps leading to the equilibrium, but not at the equilibrium. At the equilibrium only one agent has all the wealth.\par
\subsubsection{Wealth Distribution Model with Global Saving \protect\\ Propensities}
In this model, a \emph{global saving propensity} $\lambda\in(0,1)$ is introduced for the purpose of modelling the wealth distribution. The global saving propensity is the proportion of wealth that each agent will save before transacting. The range of saving propensity is taken to be between 0 and 1 implying that an agent cannot save all wealth (i.e., $\lambda=1$) or invest all wealth (i.e., $\lambda=0$). It is also assumed that the saving propensity is independent from time or any other parameters.\par
For this model, the net change of wealth is
\begin{equation}
	\Delta F_{A_{j}A_{k}}(t-\Delta t) = (1-\lambda)[\bar{\epsilon}_{A_{j}A_{k}}F(A_{j},t-\Delta t)-\epsilon_{A_{j}A_{k}}F(A_{k},t-\Delta t)]
\end{equation}
where $\epsilon_{A_{j}A_{k}}\sim U(0,1)$ with its complimentary fraction $\bar{\epsilon}_{A_{j}A_{k}}$. Under this definition, the transaction rule becomes
\begin{equation}
\label{eq:DerivingWealthModels:TransactionRuleGamma}
\begin{split}
	& F(A_{j},t) = \lambda F(A_{j},t-\Delta t)+\epsilon_{A_{j}A_{k}}(1-\lambda)(F(A_{j},t-\Delta t)+F(A_{k},t-\Delta t)) \\
	& F(A_{k},t) = \lambda F(A_{k},t-\Delta t)+\bar{\epsilon}_{A_{j}A_{k}}(1-\lambda)(F(A_{j},t-\Delta t)+F(A_{k},t-\Delta t)).
\end{split}
\end{equation}
The equilibrium of the wealth distribution dynamics in this system is described by the Gamma distribution, which is derived by Chakraborti et al. \cite{StatisticalMechanicsOfMoneySavingPropensity}. For an effective dimension $D_{\lambda}$ defined by
\begin{equation}
	\frac{D_{\lambda}}{2} = \frac{1+2\lambda}{1-\lambda},
\end{equation}
and the temperature defined by the relation
\begin{equation}
\label{eq:DerivingWealthModels:EquationEquipartitionTheorem}
	T_{\lambda} = \frac{a\langle F\rangle}{D_{\lambda}} = \frac{1-\lambda}{1+2\lambda}\langle F\rangle
\end{equation}
through the equipartition theorem, the probability density function of ``reduced'' wealth $\xi(A_{j})=\frac{F(A_{j})}{T_{\lambda}}$ is
\begin{equation}
	g(\xi(A_{j})) = \frac{1}{\Gamma(D_{\lambda}/2)}\xi(A_{j})^{\frac{D_{\lambda}}{2}-1}e^{-\xi(A_{j})} = \gamma_{\frac{D_{\lambda}}{2}}(\xi(A_{j})).
\end{equation}
$\langle F\rangle$ is the average wealth of an agent in the economy, and $\gamma_{\frac{D_{\lambda}}{2}}$ is the Gamma distribution of order $\frac{D_{\lambda}}{2}$.\par
The saving propensity in the range $\lambda\in(0,1)$ constraints the effective dimension to $2<D_{\lambda}$. For integer and half-integer values of the shape parameter $\frac{D_{\lambda}}{2}$, the probability density $g(\xi(A_{j}))$ becomes Maxwell-Boltzmann distribution at the \emph{temperature} $T_{\lambda}$ in a $D_{\lambda}$-dimensional space \cite{KineticTheoryModelsForTheDistributionOfWealth}.\par
The temperature $T_{\lambda}$ in this model describes the fluctuation of agent's wealth around the average value $\langle F\rangle$. By the equipartition theorem (i.e., equation \ref{eq:DerivingWealthModels:EquationEquipartitionTheorem}), if the saving propensity monotonically increases, the temperature monotonically decreases implying agents' wealth fluctuates less with higher $\lambda$ according to the Gamma distribution.\par
Therefore, for an economy where agents are transacting randomly, but with only a predefined proportion of their wealth, the wealth distribution will depend on the saving propensity, and the shape of the wealth distribution density function will depend on the shape parameter $\frac{D_{\lambda}}{2}$ of Gamma distribution. The shape parameter defines the skewness ($\tilde{\mu}_{3}=\frac{2\sqrt{2}}{\sqrt{D_{\lambda}}}$) and excess kurtosis ($\mathcal{K}=\frac{12}{D_{\lambda}}$) of the resulting wealth distribution; note that $1<\frac{D_{\lambda}}{2}$ because of the choice of the range for the saving propensities $\lambda\in(0,1)$, the wealth distribution will always have a unimodal shape. As the shape parameter $\frac{D_{\lambda}}{2}\rightarrow\infty$ (i.e., $\lambda\rightarrow1$), the final wealth distribution converges to the normal distribution with mean $\frac{D_{\lambda}}{2}$, which implies that the wealth is equitably distributed between the agents in the economy around the mean value, and this is consistent with the fact that as the saving propensity goes up, the agents in the economy risk less wealth in the random transactions which leads to less wealth ``redistribution''. On the other hand, for relatively low values of the shape parameter $\frac{D_{\lambda}}{2}\searrow0$ (i.e., $\lambda\searrow0$) the final wealth distribution will have positive skewness and kurtosis, leading to a wealth distribution to be left-leaning and implying that the majority of agents in the economy will be ``poor'' while there will also be a few very ``rich'' agents. This is consistent with our expectation that small saving propensity $\lambda$ results in significant wealth ``redistribution'' and creation of ``rich'' and ``poor'' agents in the economy.\par
\subsubsection{Wealth Distribution Model with Individual Saving \protect\\ Propensities}
The shortcoming of the wealth distribution model with the global saving propensity is the assumption that every agent's saving propensity is the same. In the real-world, agents are likely to have different risk tolerance, which means their saving propensities will differ. Therefore, a wealth distribution with \emph{individual saving propensities} is introduced, where every agent $A_{j}$ has his own saving propensity $\lambda_{j}$. Just like in the previous model, the saving propensities $\lambda_{j}$ are assumed to be constant.\par
The net change in wealth is defined as
\begin{equation}
	\Delta F_{A_{j}A_{k}}(t-\Delta t) = \bar{\epsilon}_{A_{j}A_{k}}(1-\lambda_{j})F(A_{j},t-\Delta t)-\epsilon_{A_{j}A_{k}}(1-\lambda_{k})F(A_{k},t-\Delta t),
\end{equation}
and the associated transaction rule is
\begin{equation}
\begin{split}
	& F(A_{j},t) = \lambda_{j}F(A_{j},t-\Delta t)+\epsilon_{A_{j}A_{k}}[(1-\lambda_{j})F(A_{j},t-\Delta t)+(1-\lambda_{k})F(A_{k},t-\Delta t)] \\
	& F(A_{k},t) = \lambda_{k}F(A_{k},t-\Delta t)+\bar{\epsilon}_{A_{j}A_{k}}[(1-\lambda_{j})F(A_{j},t-\Delta t)+(1-\lambda_{k})F(A_{k},t-\Delta t)].
\end{split}
\end{equation}
This transaction rule leads to the equilibria being described by different probability densities, and the choice of the probability density depends on the configuration of the individual saving propensities $\lambda_{j}$. However, if a system is simulated many times with different individual saving propensity configurations, the average across the equilibria that the system has attained is distributed according to the Pareto distribution \cite{MoneyInGasLikeMarkets}
\begin{equation}
	g(F(A_{j})) =
		\begin{cases}
			\frac{\alpha(F_{min})^{\alpha}}{F(A_{j})^{\alpha+1}} \quad F(A_{j})\geq F_{min} \\
			0 \quad F(A_{j})<F_{min}
		\end{cases},
\end{equation}
where the parameter $\alpha$ (i.e., \emph{Pareto exponent}) is set to 1, while $F_{min}$ is the minimum wealth an agent in the economy can have (e.g., in an economy where we cannot create  credit/debit - we have no  debt, asserting that $0\leq F(A_{j})$, so the parameter $F_{min}$ can be set to zero despite the assumption that wealth is always positive as the probability that $F(A_{j})=0$ is negligible).\par
This is an interesting result since the Pareto exponent $\alpha=\frac{log_{10}5}{log_{10}4}\approx1.161$ leads to the Pareto principle else known as the 80-20 law. This principle states that the 80\% of wealth belongs to the 20\% of agents, which coincides with the empirical data in the real-world economies \cite{HumanDevelopmentReport}.\par
For the purpose of wealth distribution modelling, this result implies that given an economy where agents are free to choose what fraction of their wealth they want to transact away randomly, we will find that after a number of iterations the wealth in the economy will be distributed according to the Pareto principle. However, this result is subject to the uniform distribution of risk averseness (i.e., saving propensities) between the agents, and is subject to the number of risk takers (i.e., agents with low individual saving propensity) in the economy.\par

\section{Dynamic Money Supply Extension}
\label{sec:DynamicMoneySupply}
In section \ref{sec:RelatedWork:DeTEcT} we mentioned that we assume the maximum supply to be constant, and for some real-world token economies (e.g., Bitcoin \cite{Bitcoin}) this assumption is satisfied. However, there are plenty of token economies that do not have a capped maximum supply and can, in practice, mint infinite number of tokens (e.g., Ethereum \cite{Ethereum}).
\subsection{Dynamic Money Supply Models}
To model wealth distribution in token economies where maximum supply is not constant, we must define an extension to our framework. Maximum supply can increase or decrease, and in the context of this paper we refer to these processes as incrementation and decrementation respectively. We define some of the possible models of dynamic maximum supply in Table \ref{table:DynamicMoneySupply:ModelTaxonomy} and will elaborate on these below (note that the taxonomy is not exhaustive as there can be infinite number of functions that model the change of maximum supply).\par
\begin{table}[h!]
\centering
\begin{tabular}{ |c|c|c|c| } 
\hline
\textbf{Name} & \textbf{Change in $M$} & \textbf{Deterministic} & \textbf{Monotonic} \\
\hline
\begin{tabular}{@{}c@{}c@{}}Simple \\ Incrementation \\ (Decrementation)\end{tabular} & \begin{tabular}{@{}c@{}}Linear Increase \\ (Decrease)\end{tabular} & Yes & Yes \\
\hline
\begin{tabular}{@{}c@{}c@{}}Compound \\ Incrementation \\ (Decrementation)\end{tabular} & \begin{tabular}{@{}c@{}}Exponential Increase \\ (Decrease)\end{tabular} & Yes & Yes \\
\hline
\begin{tabular}{@{}c@{}c@{}}Stochastic \\ Incrementation \\ (Decrementation)\end{tabular} & \begin{tabular}{@{}c@{}}Average Increase \\ (Decrease)\end{tabular} & No & No \\
\hline
\end{tabular}
\caption{Taxonomy of Dynamic Maximum Supply}
\label{table:DynamicMoneySupply:ModelTaxonomy}
\end{table}
\subsubsection{Deterministic Maximum Supply Models}
We start by looking at the simplest model from Table \ref{table:DynamicMoneySupply:ModelTaxonomy}, which is the simple incrementation and decrementation. This model linearly changes the maximum supply such that at every time step $t\in T$ the maximum supply increases or decreases by a constant value. This maximum supply model causes the maximum supply to be in the range $(-\infty,\infty)$.\par
In our framework, we can't have a zero maximum supply as the generalized equation \ref{eq:RelatedWork:DSGeneralizedConstantParams} will be undefined. Moreover, negative maximum supply doesn't make sense in practice, so we must introduce a constraint on what the input rate of incrementation (decrementation) we can set in order to ensure the range of maximum supply stays in the $(0,\infty)$ range. Below we define the simple incrementation (decrementation model).\par
\begin{definition}
\label{DefinitionSimpleIncrementation}
	\textbf{Simple incrementation (decrementation)} is the incrementation (decrementation) of maximum supply with the same amount of wealth at every time step. At time $t$, maximum supply with simple incrementation is the mapping $M:T\rightarrow(0,\infty)$ such that
	\begin{equation}
		M(t) = (1+rt)M_{initial},
	\end{equation}
where $M_{initial}\in\mathbb{R}_{>0}$ is the initial maximum supply at time $t_{initial}$, and $r\in(-\frac{1}{t},\infty)\subseteq\mathbb{R}$ for all $t\in T$ is the incrementation (decrementation) rate.
\end{definition}
The compound incrementation (decrementation) model changes the maximum supply by adding (subtracting) a constant percentage of the maximum supply from the previous time step  $t-\Delta t$. Over time, this results in the exponential rise (fall) in maximum supply, with the its range being $(-\infty,\infty)$. Yet again, we must ensure that the range of maximum supply in this model remains in the ``reasonable'' $(0,\infty)$ range by readjusting the rate parameter of this model.\par
\begin{definition}
\label{DefinitionCompoundIncrementation}
	\textbf{Compound incrementation (decrementation)} is the incrementation (decrementation) of maximum supply with the exponentially growing (shrinking) amount of wealth at every time step. At time $t$, maximum supply with compound incrementation is the mapping $M:T\rightarrow(0,\infty)$ such that
	\begin{equation}
		M(t) = (1+r)^{t}M_{initial},
	\end{equation}
where $M_{initial}\in\mathbb{R}_{>0}$ is the initial maximum supply at time $t_{initial}$, and $r\in(-1,1)\subseteq\mathbb{R}$ is the incrementation (decrementation) rate.
\end{definition}
\subsubsection{Stochastic Maximum Supply Models}
The stochastic incrementation and decrementation model introduces a stochastic process $\{R_{t}\}_{t\in T}$ that changes the maximum supply in time. These changes are not monotonic but we consider that there exists a time series of expected values for this process.\par
\begin{definition}
	\textbf{Stochastic incrementation (decrementation)} is the average incrementation (decrementation) of maximum supply with the discrete-time stochastic process $\{R_{t}\}_{t\in T}$ with the property that $0<R_{t}$ for every $t\in T$ (e.g., geometric Brownian motion). At time $t$, maximum supply with stochastic incrementation is the mapping $M:T\rightarrow(0,\infty)$ such that
	\begin{equation}
		M(t) = R_{t}M_{initial},
	\end{equation}
where $M_{initial}\in\mathbb{R}_{>0}$ is the initial maximum supply at time $t_{initial}$, and $R_{t}$ is the stochastic process with a range $\in(0,\infty)$ and the expected value time series $E(R_{t})=\mu_{R}(t)$, such that under a transformation $t\rightarrow t^{\prime}$ the expected value transforms as $\mu_{R}(t)\rightarrow\mu_{R}(t^{\prime})$.
\end{definition}
\subsection{The Time Translation Symmetry and The Discount Factor}
In physical systems, the concept of continuous symmetry, or invariance, of the system is deeply associated with the conservation laws, and in particular, the Noether's theorem \cite{NoethersTheorem} states that for every continuous symmetry that a physical system exhibits, it must have a corresponding conservation law. This statement can be reformulated to state that if a system has a symmetry, then there will be a quantity, or quantities, that are conserved (e.g., in most physical systems, energy is the conserved quantity with respect to time translation symmetry).\par
The consequence of the Noether's theorem is that there must exist a vector $J$ (i.e., \emph{conserved current}) associated with the conserved quantity, that satisfies $\frac{\partial}{\partial t}J+\nabla J=0$ (i.e., \emph{continuity equation}). The conserved current becomes a property of the system, which can aid in solving the system or verifying a proposed solution.\par
In finance and economics, we implicitly use time translation symmetry to compare values within the time series. For example, portfolio theory defines a concept of stochastic discounting factor (SDF), which allows us to connect time series of returns. In discounting models, we use a discounting rate to compare nominal cashflows values that are separated in time. In general, we often employ a discounting factor to provide us with a point of reference that allows us to compare values of the time series, and connect them from one time step to another.\par
In the context of our framework, we would like to explicitly demonstrate that under certain conditions, the discounting factor emerges as the property of the modelled system, rather than an assumed definition.The objective of this section is to demonstrate that an economic system whose maximum supply can be described using the models above, has a \emph{time translation symmetry} and that this symmetry gives rise to the conservation principle stated in the equations \ref{eq:RelatedWork:ConstantMaxSupply} and \ref{eq:RelatedWork:ConservationPrinciple}.\par
By time translation symmetry we mean that there exists an invariant quantity that will be conserved for any point $t\in T$. In equations \ref{eq:RelatedWork:ConstantMaxSupply} and \ref{eq:RelatedWork:ConservationPrinciple}, the time translation symmetry is manifested through having no discounting factor as the maximum supply $M$ is constant. Given the maximum supply models defined in the section above we can demonstrate that these maximum supply models satisfy the conservation law, and we will also show that for a generic function satisfying certain constraints, the economic system will also have a time translation symmetry, and therefore, a discounting factor.\par
\begin{theorem}[\textbf{Time Translation Symmetry}]
\label{TimeTranslationSymmetry}
	If maximum supply is constant or has simple, compound, or stochastic incrementation (decrementation), there exists a time translation symmetry in the economy.
\end{theorem}
\begin{proof}
Let the set of $n$ agent categories be
\begin{equation}
	E_{t}=\{A_{1},...,A_{n}\}, \quad n\in\mathbb{N}.
\end{equation}
Without loss of generality, assume that at $t_{initial}$ the wealth is distributed between the agent categories such that
	\begin{equation}
		\sum_{A_{j}\in E_{t}}F(A_{j},t_{initial}) = \sum^{n}_{j=1}F(A_{j},t_{initial}) = M_{initial}, \quad M_{initial}\in\mathbb{R}_{>0},
	\end{equation}
where in the first equality, it is assumed that $\lvert E_{t}\rvert=n$, and $t$ is a label for a pseudo-partition, and not the number of agent categories.\par
	For every $t\in T$ and an interval $\Delta t$, next time iteration is defined as $t^{\prime}=t+\Delta t$. Proceed for each of the cases.
\begin{description}
	\item{\textbf{Case 1}: Constant Maximum Supply}\par
	Assume for arbitrary $t\in T$ the following holds:
	\begin{equation}
	\label{TimeTranslationSymmetry11}
		\sum^{n}_{j=1}F(A_{j},t) = M_{initial}.
	\end{equation}
	It is sufficient to prove that this holds for time $t^{\prime}$.\par
	By discrete differentiation, the differential form of the equation \ref{TimeTranslationSymmetry11} is
	\begin{equation}
	\label{TimeTranslationSymmetry12}
		\sum^{n}_{j=1}\frac{F(A_{j},t)-F(A_{j},t-\Delta t)}{\Delta t} = 0.
	\end{equation}
	By translating in time $t\rightarrow t^{\prime}$ and Taylor expanding up to order $\mathcal{O}(\Delta t^{2})$,
	\begin{equation}
	\begin{split}
		\sum^{n}_{j=1}F(A_{j},t^{\prime}) &= \sum^{n}_{j=1}(F(A_{j},t)+\Delta t\frac{F(A_{j},t)-F(A_{j},t-\Delta t)}{\Delta t}) =\\
		&= \sum^{n}_{j=1}F(A_{j},t) = M_{initial},
	\end{split}	
	\end{equation}
	where for the second equality  the property in the equation \ref{TimeTranslationSymmetry12} was used. Therefore, the statement in the equation \ref{TimeTranslationSymmetry11} is proved by induction for all $t\in T$.
	\item{\textbf{Case 2}: Maximum Supply with Simple Incrementation (Decrementation)}\par
	Given the definition \ref{DefinitionSimpleIncrementation}, assume for arbitrary $t\in T$ and $r\in(-\frac{1}{\tau},\infty)\subseteq\mathbb{R}\text{ }\forall\tau\in T$ the following holds:
	\begin{equation}
	\label{TimeTranslationSymmetry21}
		\sum^{n}_{j=1}F(A_{j},t) = (1+rt)M_{initial}.
	\end{equation}
	It is sufficient to prove that this holds for time $t^{\prime}$.\par
	By discrete differentiation, the differential form of the equation \ref{TimeTranslationSymmetry21} is
	\begin{equation}
	\label{TimeTranslationSymmetry22}
		\sum^{n}_{j=1}\frac{F(A_{j},t)-F(A_{j},t-\Delta t)}{\Delta t} = rM_{initial}.
	\end{equation}
	By translating in time $t\rightarrow t^{\prime}$ and Taylor expanding up to order $\mathcal{O}(\Delta t^{2})$,
	\begin{equation}
	\begin{split}
		\sum^{n}_{j=1}F(A_{j},t^{\prime}) &= \sum^{n}_{j=1}(F(A_{j},t)+\Delta t\frac{F(A_{j},t)-F(A_{j},t-\Delta t)}{\Delta t}) =\\
		&= \sum^{n}_{j=1}F(A_{j},t)+\Delta trM_{initial} =\\
		&= (1+rt)M_{initial}+\Delta trM_{initial} =\\
		&= (1+rt^{\prime})M_{initial},
	\end{split}
	\end{equation}
	where for the second equality the property in the equation \ref{TimeTranslationSymmetry22} was used. Therefore, the statement in the equation \ref{TimeTranslationSymmetry21} is proved by induction for all $t\in T$.
	\item{\textbf{Case 3}: Maximum Supply with Compound Incrementation (Decrementation)}\par
	Given the definition \ref{DefinitionCompoundIncrementation}, assume for arbitrary $t\in T$ and $r\in(-1,1)\subseteq\mathbb{R}$ the following holds:
	\begin{equation}
	\label{TimeTranslationSymmetry31}
		\sum^{n}_{j=1}F(A_{j},t) = (1+r)^{t}M_{initial}.
	\end{equation}
	It is sufficient to prove that this holds for time $t^{\prime}$.\par
	By discrete differentiation, the differential form of the equation \ref{TimeTranslationSymmetry31} is
	\begin{equation}
	\label{TimeTranslationSymmetry32}
		\sum^{n}_{j=1}\frac{F(A_{j},t)-F(A_{j},t-\Delta t)}{\Delta t} = ln(1+r)(1+r)^{t}M_{initial}.
	\end{equation}
	By translating in time $t\rightarrow t^{\prime}$ and Taylor expanding up to order $\mathcal{O}(\Delta t^{2})$,
	\begin{equation}
	\begin{split}
		\sum^{n}_{j=1}F(A_{j},t^{\prime}) &= \sum^{n}_{j=1}(F(A_{j},t)+\Delta t\frac{F(A_{j},t)-F(A_{j},t-\Delta t)}{\Delta t}) =\\
		&= \sum^{n}_{j=1}F(A_{j},t)+\Delta tln(1+r)(1+r)^{t}M_{initial} =\\
		&= (1+r)^{t}M_{initial}+\Delta tln(1+r)(1+r)^{t}M_{initial} =\\
		&= (1+rt^{\prime})M_{initial},
	\end{split}
	\end{equation}
	where for the second equality the property in the equation \ref{TimeTranslationSymmetry32} was used. For the last equality the Taylor expansion of an exponential in general form was used
	\begin{equation}
		a^{x} = e^{ln(a)x} = 1+\frac{xln(a)}{1!}+\mathcal{O}((xln(a))^{2}),
	\end{equation}
	since $r<1$ implying that $ln(1+r)<1$.\par
	Therefore, the statement in the equation \ref{TimeTranslationSymmetry31} is proved by induction for all $t\in T$.
	\item{\textbf{Case 4}: Maximum Supply with Stochastic Incrementation (Decrementation)}\par
	Assume for arbitrary $t\in T$ and stochastic process $\{R_{t}\}_{t\in T}$ the following holds:
	\begin{equation}
	\label{TimeTranslationSymmetry41}
		\sum^{n}_{j=1}F(A_{j},t) = R_{t}M_{initial}.
	\end{equation}
	Expected value is used to stochastically discount the wealth function from one time step to the next,
	\begin{equation}
		E[\sum^{n}_{j=1}F(A_{j},t)] = \sum^{n}_{j=1}F(A_{j},t) = E[R_{t}M_{initial}] = E[R_{t}]M_{initial} = \mu_{R}(t)M_{initial}.
	\end{equation}
	Assuming that under time translations $t\rightarrow t^{\prime}$ the expected value of time series transforms as $\mu_{R}(t)\rightarrow\mu_{R}(t^{\prime})$, it is sufficient to prove that the equality \ref{TimeTranslationSymmetry41} holds for time $t^{\prime}$.\par
	By discrete differentiation, the differential form of the equation \ref{TimeTranslationSymmetry41} is
	\begin{equation}
	\label{TimeTranslationSymmetry42}
		\sum^{n}_{j=1}\frac{F(A_{j},t)-F(A_{j},t-\Delta t)}{\Delta t} = \frac{\Delta\mu_{R}(t)}{\Delta t}M_{initial}.
	\end{equation}
	By translating in time $t\rightarrow t^{\prime}$ and Taylor expanding up to order $\mathcal{O}(\Delta t^{2})$,
	\begin{equation}
	\begin{split}
		\sum^{n}_{j=1}F(A_{j},t^{\prime}) &= \sum^{n}_{j=1}(F(A_{j},t)+\Delta t\frac{F(A_{j},t)-F(A_{j},t-\Delta t)}{\Delta t}) =\\
		&= \sum^{n}_{j=1}F(A_{j},t)+\Delta t\frac{\Delta\mu_{R}(t)}{\Delta t}M_{initial} =\\
		&= \mu_{R}(t)M_{initial}+\Delta t\frac{\Delta\mu_{R}(t)}{\Delta t}M_{initial} =\\
		&= \mu_{R}(t^{\prime})M_{initial},
	\end{split}
	\end{equation}
	where for the second equality the property in the equation \ref{TimeTranslationSymmetry42} was used. For the last equality the Taylor expansion approximation of $\mu_{R}(t)$ was used.\par
	Therefore, the statement in the equation \ref{TimeTranslationSymmetry41} is proved by induction for all $t\in T$.
\item{\textbf{General Case}: General form of the incrementation function.}\par
	Assume for arbitrary $t\in T$ the following holds:
	\begin{equation}
	\label{TimeTranslationSymmetry51}
		\sum^{n}_{j=1}F(A_{j}, t) = g(t),
	\end{equation}
	where $g(t)$ is an infinitely differentiable maximum supply function such that $g(t_{initial})=M_{initial}=g(t_{0})$, and $g(t)$ has a Taylor series approximation. It is sufficient to prove that the equality above holds for time $t^{\prime}$.\par
	By discrete differentiation, the differential form of the equation \ref{TimeTranslationSymmetry51} is
	\begin{equation}
	\label{TimeTranslationSymmetry52}
		\sum^{n}_{j=1}\frac{F(A_{j},t)-F(A_{j},t-\Delta t)}{\Delta t} = \frac{\Delta}{\Delta t}[g(t)] = \frac{g(t)-g(t-\Delta t)}{\Delta t}.
	\end{equation}
	By translating in time $t\rightarrow t^{\prime}$ and Taylor expanding up to order $\mathcal{O}(\Delta t^{2})$,
	\begin{equation}
	\begin{split}	
		\sum^{n}_{j=1}F(A_{j},t^{\prime}) &= \sum^{n}_{j=1}(F(A_{j},t)+\Delta t\frac{F(A_{j},t)-F(A_{j},t-\Delta t)}{\Delta t}) =\\
		&= \sum^{n}_{j=1}F(A_{j}, t)+\Delta t\frac{g(t)-g(t-\Delta t)}{\Delta t} =\\
		&= g(t)+\Delta t\frac{\Delta}{\Delta t}[g(t)] =\\
		&= g(t^{\prime}),
	\end{split}
	\end{equation}
	where we used the infinitely differentiable property of $g(t)$, where its Taylor approximation around $t^{\prime}$ is
	\begin{equation}
		g(t^{\prime}) = g(t) + \Delta t \frac{\Delta}{\Delta t}[g(t^{\prime})].
	\end{equation}
	Therefore, the statement in the equation \ref{TimeTranslationSymmetry51} is proved by induction for all $t\in T$.
\end{description}
\end{proof}
This theorem demonstrates that an economy where maximum supply function is infinitely differentiable the economy will have a time translation symmetry. This result allows us to make a couple of interesting statements about economies modelled with DeTEcT.\par
First, the time translation symmetry allows us to connect the values in the time series of wealth distribution by discounting changes in the maximum supply. By proving the theorem above, we have demonstrated that an economy whose maximum supply is described by the given supply models, the notion of the discounting factor is promoted from the definition to the property of the system. For example, if we assume that an economy has a simple incrementation (decrementation), then the sum of wealth functions at a given time step $t\in T$ is
\begin{equation}
	\frac{1}{1+rt}\sum^{n}_{j=1}F(A_{j},t) = M_{initial},
\end{equation}
where $\frac{1}{1+rt}$ is just the discounting factor at this time step. Therefore, in our framework, the time-value of money is the phenomenon caused by the symmetry associated with the maximum supply evolution, and is derived from the properties of the economy being modelled.\par
Second, we note that we can now add the incrementation term to the general equation of the dynamical system in equation \ref{eq:RelatedWork:DSGeneralizedConstantParams} such that it reflects the dynamic money supply.\par
For us to add the general maximum supply function, $g(t)$, we must assume that $g(t)$ satisfies $g(t_{initial})=M_{initial}=g(t_{0})$, and we also must define incrementation on an ``agent category''-basis, as the incrementation mechanism tells us how much new wealth has been created. But in order to add incrementation to the dynamical system we must specify what agent categories gain (lose) wealth (in most cases, we can assume that the control mechanism will be the only agent category to gain or lose wealth due to the dynamic maximum supply).\par
Let $G(A_{j},t)$ be the amount by which the wealth of agent category $A_{j}$ has been changed at time step $t$. It is required that the sum of changes in wealth of all agent categories at time step $t$ is the difference between the maximum supply of the economy between the last and the current time step,
\begin{equation}
	\sum^{n}_{j=0}G(A_{j},t) = g(t)-g(t-\Delta t).
\end{equation}
We can define a vector of wealth changes, $\vv{G}(t)$, where every entry is the incrementation (decrementation) of wealth of the respective agent category at the given time step,
\begin{equation}
	\vv{G}(t) = (G(A_{1},t), ..., G(A_{n},t))^{T}, \quad A_{1},...,A_{n}\in E_{t}.
\end{equation}
Now, we add the incrementation (decrementation) mechanism to the dynamical system in equation \ref{eq:RelatedWork:DSGeneralizedConstantParams} to obtain its modified version,
\begin{equation}
\label{eq:DynamicMoneySupply:DSGeneralized}
	\frac{\Delta}{\Delta t}[\vv{F}(t)] = \frac{1}{g(t)}\vv{F}(t)\odot[\mathcal{B}\cdot\vv{F}(t)]+\Gamma\cdot\vv{F}(t)+\vv{G}(t) \quad t\in T.
\end{equation}
We can prove that by induction the equation \ref{eq:DynamicMoneySupply:DSGeneralized} holds for future time steps. First, assume that the equation holds for $t_{0}$ and for some arbitrary $t\in T$. Then, for the next time step $t^{\prime}=t+\delta t$, the Taylor expansion of $\vv{F}(t^{\prime})$ is
\begin{equation}
\begin{split}
	\vv{F}(t^{\prime}) &= \vv{F}(t)+\delta t\frac{\Delta}{\Delta t}[\vv{F}(t)]+\mathcal{O}(\delta t^{2}) \approx\\
	&\approx \vv{F}(t) + \delta t[\frac{1}{g(t)}\vv{F}(t)\odot[\mathcal{B}\cdot\vv{F}(t)]+\Gamma\cdot\vv{F}(t)+\vv{G}(t)],
\end{split}
\end{equation}
where we used equation \ref{eq:DynamicMoneySupply:DSGeneralized}, and ignored all terms of $\mathcal{O}(\delta t^{2})$ and higher, as we consider the interval $\delta t$ to be very small.\par
The finite difference of $\vv{F}$ around $t^{\prime}$ is
\begin{equation}
\begin{split}
	\frac{\Delta t}{\Delta t}[\vv{F}(t^{\prime})] &= \frac{\Delta}{\Delta t}[\vv{F}(t)]+\frac{\Delta}{\Delta t}[\delta t\frac{\Delta}{\Delta t}[\vv{F}(t)]] =\\
	&= \frac{1}{g(t)}\vv{F}(t)\odot[\mathcal{B}\cdot\vv{F}(t)]+\Gamma\cdot\vv{F}(t)+\vv{G}(t) +\\
	&+ \delta t\frac{\Delta}{\Delta t}[\frac{1}{g(t)}\vv{F}(t)\odot[\mathcal{B}\cdot\vv{F}(t)]+\Gamma\cdot\vv{F}(t)+\vv{G}(t)] =\\
	&= \frac{1}{g(t)}\vv{F}(t)\odot[\mathcal{B}\cdot\vv{F}(t)]+\delta t\frac{\Delta}{\Delta t}[\frac{1}{g(t)}\vv{F}(t)\odot[\mathcal{B}\cdot\vv{F}(t)]] +\\
	&+ \Gamma\cdot\vv{F}(t)+\delta t\frac{\Delta}{\Delta t}[\Gamma\cdot\vv{F}(t)]+\vv{G}(t)+\delta t\frac{\Delta}{\Delta t}[\vv{G}(t)] =\\
	&= \frac{1}{g(t^{\prime})}\vv{F}(t^{\prime})\odot[\mathcal{B}\cdot\vv{F}(t^{\prime})]+\Gamma\cdot\vv{F}(t^{\prime})+\vv{G}(t^{\prime}),
\end{split}
\end{equation}
which proves that the equation \ref{eq:DynamicMoneySupply:DSGeneralized} holds by induction for all time steps. This equation is the general form of the dynamical system with static parameters and dynamic money supply.
\subsection{Dynamic Money Supply and Parametrization Techniques}
At last, we address the case where dynamic money supply and different parametrization techniques are applied simultaneously. The dynamical system with dynamic maximum supply and dynamic parameters will have the same form as the system in the section above, but with the explicit time-dependence of $\mathcal{B}(t)$ and $\Gamma(t)$,
\begin{equation}
	\frac{\Delta}{\Delta t}[\vv{F}(t)] = \frac{1}{g(t)}\vv{F}(t)\odot[\mathcal{B}(t)\cdot\vv{F}(t)]+\Gamma(t)\cdot\vv{F}(t)+\vv{G}(t) \quad t\in T.
\end{equation}
It can be proven, that this form of the dynamical system stands for the future time step using the same methodology as described in the previous section, but we will alleviate the calculation for the purpose of the paper.

\section{Conclusion}
The aim of this paper was to demonstrate how our proposed framework, DeTEcT \cite{DeTEcT}, fits into theoretical research on wealth distribution models, and how it can be improved to be more flexible for modelling wider range of real-world token economies with different features.\par
In this paper, we described multiple ways that our framework can be parametrized to remain very flexible. Static probabilistic parametrization enables us to model the wealth distribution dynamics based on the real-world data, and dynamic deterministic parametrization can be used to simulate dynamically-changing set of policies in the economy, while static deterministic and dynamic probabilistic parametrizations are convenient for simulating different economic scenarios with static or variable set of policies. In the forward propagation, the dynamic parametrization techniques can be customized further with the choice of deterministic functions, or stochastic processes that define demands and prices for individual goods.\par
Additionally, we introduced a dynamic money supply extension that covers token economies with time-dependent money supply (e.g., Ethereum \cite{Ethereum}). This extension works harmoniously with the parametrization techniques and expands the use cases for DeTEcT.\par
In summary, we have added improvements to our framework, which will help us with building the simulation engine that performs analysis and runs simulations of wealth distribution in real-world token economies, and will allow us to study the interactions between different agents and agent categories. Despite these improvements, this paper does not present an exhaustive list of modifications and features that can be added to DeTEcT (e.g., tokens with expiration mechanism), but we believe that the extensions we have presented here are the ones that carry the most significance for the framework, and demonstrate the approaches in which researchers can modify the framework to use it in different context.


\begin{thebibliography}{99}
	\bibitem{DeTEcT} R. Sadykhov, G. Goodell, D. de Montigny, M. Schoernig and P. Treleaven, December 2023. \emph{Decentralized Token Economy Theory (DeTEcT)}. \url{https://doi.org/10.3389/fbloc.2023.1298330}
	
	\bibitem{ISOVocabulary} \emph{Blockchain and Distributed Ledger Technologies - Vocabulary}, January 2023. \url{https://www.iso.org/obp/ui/#iso:std:iso:22739:dis:ed-2:v1:en}
	
	\bibitem{ISOVocabularyEconomics} \emph{Information Technology - Security Techniques - Information Security Management - Organizational Economics}. \url{https://www.iso.org/obp/ui/#iso:std:iso-iec:tr:27016:ed-1:v1:en:term:3.7}
	
	\bibitem{WealthOfNations} A. Smith, March 1776. \emph{The Wealth of Nations}, Wordsworth Editions, 2012.
	
	\bibitem{StatisticalMechanicsOfMoney} A. Dragulescu and V.M. Yakovenko, 2000. \emph{Statistical Mechanics of Money}, The European Physical Journal B 17, pages 723-729, 2000. \url{https://doi.org/10.1007/s100510070114}
	
	\bibitem{DistributionsOfMoneyInModelMarketsOfEconomy} A. Chakraborti, 2002. \emph{Distributions of Money in Model Markets of Economy}, Int. J. Mod. Phys. C13, 1315, 2002. \url{https://arxiv.org/abs/cond-mat/0205221v1}
	
	\bibitem{StatisticalMechanicsOfMoneySavingPropensity} A. Chakraborti and B.K. Chakrabarti, 2000. \emph{Statistical Mechanics of Money: How Saving Propensity Affects Its Distribution}, The European Physical Journal B 17, pages 167-170, 2000. \url{https://doi.org/10.1007/s100510070173}
	
	\bibitem{MoneyInGasLikeMarkets} A. Chatterjee, B.K. Chakrabarti and S.S. Manna, 2003. \emph{Money In Gas-Like Markets: Gibbs and Pareto Laws},Physica Scripta T106, November 2003. \url{https://arxiv.org/abs/cond-mat/0311227}
	
	\bibitem{KineticTheoryModelsForTheDistributionOfWealth} M. Patriarca, A. Chakraborti, K. Kaski and G. Germano, 2005. \emph{Kinetic Theory Models For the Distribtion of Wealth: Power Law From Overlap of Exponentials}, Econophysics of Wealth Distribution, pages 93-110, Springer, Milano, 2005. \url{https://doi.org/10.1007/88-470-0389-X_10}
	
	\bibitem{RelaxationInStatisticalManyAgentEconomyModels} M. Patriarca, A. Chakraborti, E. Heinsalu and G. Germano, 2008. \emph{Relaxation in Statistical Many-Agent Economy Models}, The European Physical Journal B 57, pages 219-224, May 2007. \url{https://doi.org/10.1140/epjb/e2007-00122-7}
	
	\bibitem{ParetoPrinciple} V. Pareto, 1896. \emph{Cours d'\'Economie Politique}, Droz, Gen\'eve, 1964.
	
	\bibitem{HumanDevelopmentReport} \emph{United Nations Human Development Report 1992}, New York: Oxford University Press, page 34, 1992. \url{https://hdr.undp.org/content/human-development-report-1992}
	
	\bibitem{NoethersTheorem} E. Noether, 1918. \emph{Invariante Variationsprobleme}, Transport Theory and Statistical Physics 1(3), Translated by M.A. Tavel, 2018. \url{https://doi.org/10.48550/arXiv.physics/0503066}
	
	\bibitem{Bitcoin} S. Nakamoto, 2008. \emph{Bitcoin: A Peer-to-Peer Electronic Cash System}.
	
	\bibitem{Ethereum} V. Buterin, 2014. \emph{Ethereum: A Next-Generation Smart Contract and Decentralized Application Platform}. \url{https://ethereum.org/669c9e2e2027310b6b3cdce6e1c52962/Ethereum_Whitepaper_-_Buterin_2014.pdf}
	
\end{thebibliography}
\end{document}